\documentclass[11pt]{article}
\usepackage[margin=1in]{geometry}
\usepackage{datetime}
\usepackage{enumitem}
\usepackage{wrapfig}
\usepackage{subcaption}
\usepackage[T1]{fontenc}
\usepackage[utf8]{inputenc}
\usepackage{pgfplots}
\usepackage{amsmath, amssymb}
\usepackage{amsthm}
\usepackage{color}
\usepackage{cases}
\usepackage{xparse}
\usepackage{xargs}
\usepackage{appendix}
\usepackage[linesnumbered,lined,boxed,commentsnumbered,noend]{algorithm2e}
\usepackage{algpseudocode}
\usepackage{verbatim, xspace}
\usepackage{tikz}
\usepackage{mathtools}
\usepackage{tikzit}

\tikzstyle{vertex}=[fill=white, draw=black, shape=circle, inner sep=2pt]
\tikzstyle{new style 0}=[fill=black, draw=black, shape=circle, inner sep=0.5pt]
\tikzstyle{Triangle}=[fill={rgb,255: red,191; green,191; blue,191}, draw={rgb,255: red,128; green,128; blue,128}, regular polygon sides=3, regular polygon]

\tikzstyle{new edge style 0}=[-, draw=black, ultra thick]
\tikzstyle{red}=[-, draw=red, ultra thick]
\tikzstyle{green}=[-, draw={rgb,255: red,0; green,243; blue,0}, ultra thick]
\tikzstyle{gray}=[-, draw={rgb,255: red,177; green,177; blue,177}, dashed]
\tikzstyle{dotted}=[-, draw={rgb,255: red,128; green,128; blue,128}, densely dotted]
\tikzstyle{shaded}=[-, draw={rgb,255: red,191; green,191; blue,191}]
\tikzstyle{canonical line}=[-, thick, draw=black, preaction={{ draw,cyan,-, double=cyan, double distance=3\pgflinewidth}
}]
\tikzstyle{wave}=[decorate, -, thick, decoration=snake]
\tikzstyle{arrow}=[->, thick, dashed]
\tikzstyle{golden}=[-, dashed, draw={rgb,255: red,255; green,218; blue,9}, thick]
\tikzstyle{sheaf}=[-, preaction={{draw,yellow,-, double=yellow, double distance=3\pgflinewidth}}, thick]
\tikzstyle{blue}=[-, ultra thick, draw={rgb,255: red,15; green,123; blue,255}]
\tikzstyle{new edge style 1}=[-, ultra thick, draw={rgb,255: red,255; green,218; blue,7}]

\usepackage{hyperref}
\usepackage{pgfplots}

\usepackage{xcolor}

\newcommand{\citet}[1]{\cite{#1}}

\usepackage[colorinlistoftodos,prependcaption,textsize=tiny]{todonotes}
\newcommandx{\unsure}[2][1=]{\todo[linecolor=green,backgroundcolor=green!25,bordercolor=green,#1]{\normalsize #2}}
\newcommandx{\improvement}[2][1=]{\todo[inline,linecolor=blue,backgroundcolor=blue!05,bordercolor=blue,#1]{\normalsize #2}}
\newcommandx{\info}[2][1=]{\todo[linecolor=yellow,backgroundcolor=yellow!25,bordercolor=yellow,#1]{#2}}
\newcommandx{\floatmodel}[2][1=]{\todo[inline,linecolor=red,backgroundcolor=yellow!25,bordercolor=yellow,#1]{#2}}
\newcommandx{\thiswillnotshow}[2][1=]{\todo[disable,#1]{#2}}
\newcommandx{\celine}[2][1=]{\todo[inline,linecolor=green,backgroundcolor=green!25,bordercolor=green,caption={\normalsize \textbf{Celine}},#1]{\normalsize #2}}
\newcommandx{\karol}[2][1=]{\todo[inline,linecolor=blue,backgroundcolor=blue!25,bordercolor=blue,caption={\normalsize \textbf{Karol}},#1]{\normalsize #2}}
\newcommandx{\jesper}[2][1=]{\todo[inline,linecolor=red,backgroundcolor=red!25,bordercolor=red,caption={\normalsize \textbf{Jesper}},#1]{\normalsize #2}}
\newcommandx{\michal}[2][1=]{\todo[inline,linecolor=gray,backgroundcolor=red!25,bordercolor=red,caption={\normalsize \textbf{Micha\l{}}},#1]{\normalsize #2}}

\newtheorem{theorem}{Theorem}
\newtheorem{definition}[theorem]{Definition}
\newtheorem{lemma}[theorem]{Lemma}
\newtheorem{corollary}[theorem]{Corollary}

\newtheorem{observation}[theorem]{Observation}

\numberwithin{theorem}{section}

\newcommand{\Oh}{\mathcal{O}}
\newcommand{\Os}{\Oh^{\star}}

\newcommand{\N}{\mathbb{N}}
\newcommand{\Z}{\mathbb{Z}}
\newcommand{\Ff}{\mathcal{F}}
\newcommand{\Rr}{\mathcal{R}}
\newcommand{\Ss}{\mathcal{S}}
\newcommand{\Cc}{\mathcal{C}}
\newcommand{\Mm}{\mathcal{M}}

\newcommand{\ol}{\overline}

\newcommand{\PCC}{{\sc{Partial Cycle Cover}}\xspace}

\newcommand{\child}{\mathsf{children}}
\newcommand{\parent}{\mathsf{parent}}
\newcommand{\subtree}{\mathsf{subtree}}
\newcommand{\tail}{\mathsf{tail}}
\newcommand{\broom}{\mathsf{broom}}
\newcommand{\cc}{\ensuremath{\mathtt{cc}}}
\newcommand{\tw}{\ensuremath{t}}
\newcommand{\td}{\ensuremath{d}}

\renewcommand{\leq}{\leqslant}
\renewcommand{\geq}{\geqslant}

\newcommand{\defproblem}[3]{
	\vspace{2mm}
	\vspace{1mm}
	\noindent\fbox{
		\begin{minipage}{0.95\textwidth}
			#1 \\
			{\bf{Input:}} #2  \\
			{\bf{Task:}} #3
		\end{minipage}
	}
	\vspace{2mm}
}

\title{Hamiltonian Cycle Parameterized by Treedepth in Single Exponential Time and Polynomial Space}
\date{}

\author{
    Jesper Nederlof\footnote{Utrecht University, The
    Netherlands, \texttt{j.nederlof@uu.nl}. Supported by
    the project CRACKNP that has received funding from the European
    Research Council (ERC) under the European Union’s Horizon 2020 research and
    innovation programme (grant agreement No 853234).}
    \and
    Micha\l{} Pilipczuk\footnote{Institute of Informatics, University of
    Warsaw, Poland, \texttt{michal.pilipczuk@mimuw.edu.pl}. Supported by
    the project TOTAL that has received funding from the European
    Research Council (ERC) under the European Union’s Horizon 2020 research and
    innovation programme (grant agreement No 677651).}
    \and
    Céline M. F. Swennenhuis\footnote{Eindhoven University of Technology, The
    Netherlands, \texttt{c.m.f.swennenhuis@tue.nl}. Supported by the Netherlands
    Organization for Scientific Research under project no. 613.009.031b.}
    \and
    Karol W\k{e}grzycki\footnote{Institute of Informatics, University of
    Warsaw, Poland, \texttt{k.wegrzycki@mimuw.edu.pl}. Supported by
    the grants 2016/21/N/ST6/01468 and 2018/28/T/ST6/00084 of the Polish National
    Science Center and project TOTAL that has received funding from the European
    Research Council (ERC) under the European Union’s Horizon 2020 research and
    innovation programme (grant agreement No 677651).}
}

\begin{document}

\maketitle

\thispagestyle{empty}
For many algorithmic problems on graphs of treewidth $\tw$,
a standard dynamic programming approach gives an algorithm with time and space complexity $2^{\Oh(\tw)}\cdot n^{\Oh(1)}$.
It turns out that when one considers the more restrictive parameter treedepth,
it is often the case that a variation of this technique can be used to reduce the space complexity to polynomial, while retaining time complexity of the form $2^{\Oh(\td)}\cdot n^{\Oh(1)}$, where $\td$ is the treedepth.
This transfer of methodology is, however, far from automatic. For instance, for problems with connectivity constraints, standard dynamic programming techniques give algorithms with time and space complexity
$2^{\Oh(\tw\log \tw)}\cdot n^{\Oh(1)}$ on graphs of treewidth $\tw$, but it is not clear how to convert them into time-efficient polynomial space algorithms for graphs of low treedepth.

Cygan et al. (FOCS'11) introduced the Cut\&Count technique and showed that a certain class of problems with connectivity constraints can be solved in time and space complexity $2^{\Oh(\tw)}\cdot n^{\Oh(1)}$.
Recently, Hegerfeld and Kratsch (STACS'20) showed that, for some of those problems, the Cut\&Count technique can be also applied in the setting of treedepth, and it gives algorithms with running time $2^{\Oh(\td)}\cdot n^{\Oh(1)}$ 
and polynomial space usage. However, a number of important problems eluded such a treatment, with the most prominent examples being {\sc{Hamiltonian Cycle}} and {\sc{Longest Path}}.

In this paper we clarify the situation by showing that 
{\sc{Hamiltonian Cycle}}, {\sc{Hamiltonian Path}}, {\sc{Long Cycle}}, {\sc{Long Path}}, and {\sc{Min Cycle Cover}} all admit $5^\td\cdot n^{\Oh(1)}$-time and polynomial space algorithms on graphs of treedepth $\td$.
The algorithms are randomized Monte Carlo with only false negatives.

\begin{picture}(0,0)
\put(462,-145)
{\hbox{\includegraphics[width=40px]{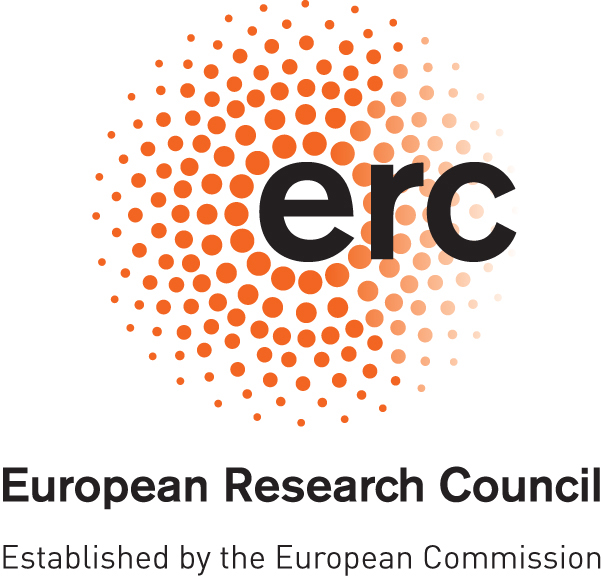}}}
\put(452,-205)
{\hbox{\includegraphics[width=60px]{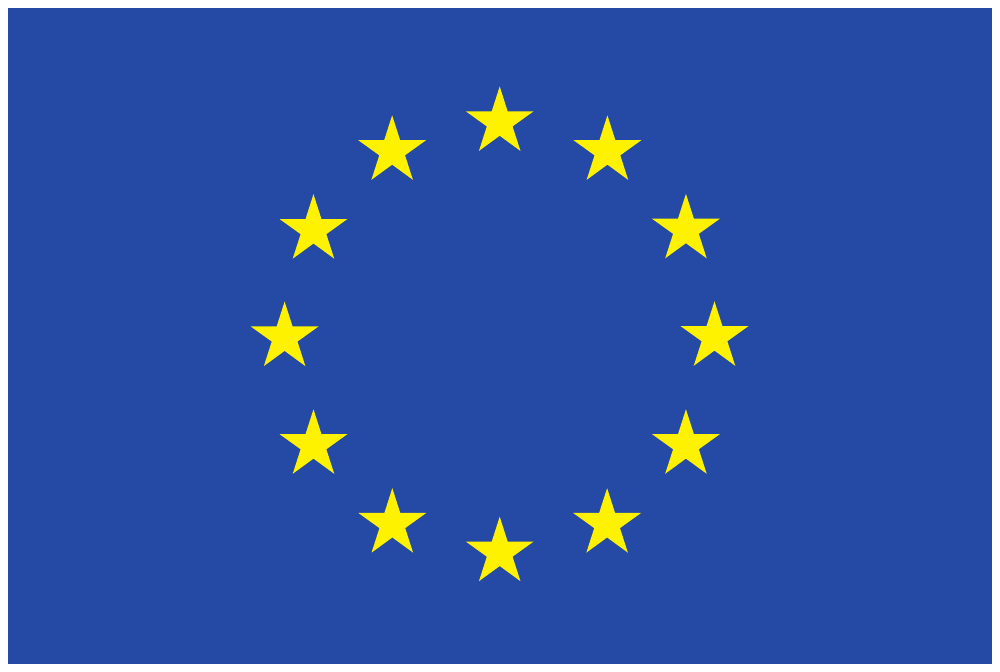}}}
\end{picture}

\clearpage
\setcounter{page}{1}

\section{Introduction}

It is widely believed that no NP-hard problem admits a polynomial time algorithm.
However, actual instances of problems that we are interested in solving often
admit much more structure than a general instance. This observation gave rise to the field
of \emph{parameterized complexity}, where the hardness of an instance does not
depend exclusively on the input size. In the parameterized regime, we assume that
each instance is equipped with an additional parameter $k$ and the goal is to
give a {\em{fixed-parameter algorithm}}: an algorithm with running time $f(k)\cdot n^{\Oh(1)}$, where $f$ is a function independent of~$n$. 
After settling that a problem admits such an algorithm, it is natural to look
for one with function $f$ as low as possible. We refer to~\cite{book1,book2,book3} for an introduction to
parameterized complexity.

One of the most widely used parameters is the \emph{treewidth} $\tw$ of the input
graph.  Usually, problems that involve only constraints of local nature admit an
algorithm with running time of the form $2^{\Oh(\tw)}\cdot n^{\Oh(1)}$~\cite{book1}. 
For a long time, such algorithms remained out of reach
for problems involving connectivity constraints, and for those only $2^{\Oh(\tw\log \tw)}\cdot n^{\Oh(1)}$-time algorithms were
known. The breakthrough came with the Cut\&Count technique, introduced by
Cygan et al. in~\citet{focs2011}, that allows one to design randomized Monte-Carlo algorithms with running times of the form $2^{\Oh(\tw)}\cdot n^{\Oh(1)}$
for a wide range of connectivity problems, e.g.,
\textsc{Hamiltonian Path}, \textsc{Connected Vertex Cover}, \textsc{Connected
Dominating Set}, etc. The technique was subsequently derandomized~\cite{rank-based,matroid-derandomization}.

One of the main issues with standard dynamic programming algorithms is that they tend to have prohibitively large space usage. The natural goal is therefore to
reduce the space complexity while not sacrificing much on the time complexity. Unfortunately, Drucker et al.~\cite{DruckerNS16} and Pilipczuk and Wrochna~\cite{pw-stacs2016}
gave some complexity-theoretical evidence that for dynamic programming on graphs of bounded treewidth, such a reduction is probably impossible. For example, they showed that under plausible assumptions,
there is no algorithm that works in time $2^{\Oh(\tw)}\cdot n^{\Oh(1)}$ and uses $2^{o(\tw)}\cdot n^{\Oh(1)}$ space for the \textsc{$3$-Coloring} or \textsc{Independent Set} problem. 

\paragraph{Treedepth.} 
The aforementioned issues motivate the research on a different, more restrictive parameterization, for which the reduction of space complexity would be possible. In this paper we will
consider the parameterization by \emph{treedepth}, defined as follows.

\begin{definition}
    An {\em{elimination forest}} of a graph $G$ is a rooted forest $F$ on the same vertex set as $G$ such that for every edge $uv$ of $G$, either $u$ is an ancestor of $v$ in $F$ or $v$ is an ancestor of $u$ in $F$.
    The {\em{treedepth}} of $G$ is the minimum possible depth of an elimination forest of $G$.
\end{definition}

The treedepth of a graph is never smaller than its treewidth, but it is also never larger than the treewidth times $\log n$.
In many concrete cases, the two parameters have the same advantages. 
For example, planar graphs have treewidth $\Oh(\sqrt{n})$, but also \emph{treedepth} $\Oh(\sqrt{n})$. 

It has been recently realized that on graphs of treedepth $d$, many algorithmic problems indeed can be solved in time $2^{\Oh(d)}\cdot n^{\Oh(1)}$ and using only polynomial space.\footnote{Throughout the introduction, when we speak about a graph of treedepth $d$, we mean a graph supplied with an elimination forest of depth $d$. 
While in the case of treewidth, a tree decomposition of approximately (up to a constant factor) optimum width can be computed in time $8^t\cdot n^{\Oh(1)}$~\cite{GM13,book1}, the existence of such an approximation algorithm 
for treedepth is a notorious open problem.}
For the most basic problems, such as \textsc{3-Coloring} and \textsc{Independent Set}, a simple branching algorithms achieves such complexity.
However, in contrast to the treewidth parameterization, for many more complex problems it is highly non-trivial, yet possible to establish similar bounds.
One technique that turns out to be useful here is
the framework of algebraic transforms introduced by Loksthanov and Nederlof~\citet{algebraization}, who
demonstrated how to reduce the space requirements of many dynamic programming
algorithms to polynomial in the input size by reorganizing the computation using a suitable transform. F\"urer and Yu~\citet{furer-1} applied this framework to give
$2^{\Oh(\td)}\cdot n^{\Oh(1)}$-time and polynomial space algorithms on graphs of treedepth $\td$ for the \textsc{Dominating Set} problem and for the problem of counting the number of perfect matchings.
Pilipczuk and Wrochna~\citet{pw-stacs2016} considered algorithms with even more restricted space
requirements: they showed that \textsc{$3$-Coloring}, \textsc{Dominating Set}, and
\textsc{Vertex Cover} admit algorithms that work in $2^{\Oh(\td)}\cdot n^{\Oh(1)}$ time and
use $\Oh(\td + \log{n})$ space. 
For \textsc{Dominating Set} they avoided the explicit use of algebraization and instead provided a more combinatorial interpretation based on what one could call \emph{inclusion-exclusion branching}.
Later, Pilipczuk and Siebertz~\citet{soda19} used
color-coding to give an $2^{\Oh(\td \log{\td})}\cdot n^{\Oh(1)}$-time and polynomial
space algorithm for the {\sc{Subgraph Isomorphism}} problem. Recently, Belbasi and F\"urer~\citet{furer-2}
presented an algorithm for counting Hamiltonian cycles in time $(4t)^d\cdot n^{\Oh(1)}$ and using
polynomial space, where $t$ is the width of a given tree decomposition and $d$ is its (suitably defined) depth.

\paragraph{Treedepth and Cut\&Count.}
Very recently, Hegerfeld and Kratsch~\citet{stacs2020} demonstrated that the Cut\&Count technique can be also applied in the setting of the treedepth parameterization.
Consequently, they gave randomized algorithms with running times
$2^{\Oh(\td)}\cdot n^{\Oh(1)}$ and polynomial space usage for a number of problems with connectivity constraints such as {\sc{Connected Vertex Cover}},
{\sc{Connected Dominating Set}}, {\sc{Feedback Vertex Set}}, or {\sc{Steiner Tree}}. 
However, Hegerfeld and Kratsch found it problematic to apply the methodology to several important problems originally considered by Cygan et al.~\cite{focs2011} in the context of Cut\&Count.
Specifically, these are problems based on selection of edges rather than vertices, such as {\sc{Hamiltonian Cycle}} or {\sc{Long Cycle}}.
For this reason, Hegerfeld and Kratsch explicitly asked in~\cite{stacs2020} whether {\sc{Hamiltonian Cycle}}, {\sc{Hamiltonian Path}}, {\sc{Long Cycle}}, and {\sc{Min Cycle Cover}} also admit $2^{\Oh(\td)}\cdot n^{\Oh(1)}$-time
and polynomial space algorithms on graphs of treedepth $\td$ (see Appendix~\ref{sec:problems} for problem definitions\footnote{Note that when discussing the {\sc{Long Path}} and the {\sc{Long Cycle}} problems, we use the
letter $\ell$ to denote the required length of a path, respectively of a cycle, instead of the letter $k$ that is perhaps more traditionally used in this context.}).




\paragraph{Our contribution.}
In this paper we introduce additional techniques that allow us to extend the
results of~\citet{stacs2020} and to answer the abovementioned open problem of Hegerfeld and Kratsch in the affirmative.
More precisely, we prove the following theorem.

\begin{theorem}
    \label{main-cor}
    There is a randomized algorithm that given a graph $G$ together with its
    elimination forest of depth $\td$, and number $k \in \N$, solves
    \textsc{Hamiltonian Cycle}, \textsc{Hamiltonian Path}, \textsc{$k$-Cycle}, \textsc{$k$-Path} and \textsc{Min Cycle Cover} in time $5^\td\cdot n^{\Oh(1)}$
    and using polynomial space. The algorithm has a one-sided error: it may give false
    negatives with probability at most $\frac{1}{2}$.
\end{theorem}

In fact, Theorem~\ref{main-cor} is an easy corollary of the following result for a generalization of the considered problems.
In the \PCC problem we are given an undirected graph $G$ and integers $k$ and $\ell$, and we ask whether in $G$ there is a family of at most $k$ vertex-disjoint cycles that jointly visit exactly $\ell$ vertices.
We will prove the following theorem.

\begin{theorem}
    \label{main-theorem}
     There is a randomized algorithm that given a graph $G$ together with its
     elimination forest of depth $\td$, and numbers $k,\ell\in \N$, solves the
     \PCC problem for $G,k,\ell$ in time $5^\td\cdot n^{\Oh(1)}$
    and using polynomial space. The algorithm has a one-sided error: it may give false
    negatives with probability at most $\frac{1}{2}$.
\end{theorem}

To see that Theorem~\ref{main-theorem} implies Theorem~\ref{main-cor}, note that
\textsc{Hamiltonian Cycle}, {\sc{Min Cycle Cover}} and {\sc{Long Cycle}} are
special cases of the \PCC (for fixed parameters $k$ and $\ell$).

To solve \textsc{Long Path}, we can simply iterate through all pairs of non-adjacent vertices $s,t$ and apply the {\sc{Long Cycle}} algorithm to the graph $G$ with edge $st$ added; this increases the treedepth
by at most $1$ and the provided elimination forest can be easily adjusted. It is easy to see that then the original graph $G$ contains a simple path on $\ell$ vertices if and only if for some choice of $s$ and $t$,
we find a cycle of length $\ell$ in $G$ augmented with the edge $st$.
Finally, \textsc{Hamiltonian Path} is just {\sc{Long Path}} applied for $\ell=|V(G)|$.

We remark that our algorithmic findings have concrete applications outside of the realm of structural parameterizations.
For instance, Lokshtanov et al.~\cite{k-path} gave a $2^{\Oh(\sqrt{\ell}\log^2 \ell)}\cdot n^{\Oh(1)}$-time polynomial space algorithm for the {\sc{Long Path}} problem on $H$-minor-free graphs, for every fixed $H$.
In Appendix~\ref{applications} we present how using our results one can improve the running time to $2^{\Oh(\sqrt{\ell}\log \ell)}\cdot n^{\Oh(1)}$ while keeping the polynomial space complexity.

\paragraph{Our techniques.}
Similarly to Hegerfeld and Kratsch~\citet{stacs2020} we use the Cut\&Count framework, 
but we apply a different new view on the Count part, suited for problems based on edge selection. 
The main idea is that instead of counting cycle covers, as a standard application of Cut\&Count would do,
we count perfect matchings in an auxiliary graph, constructed by replacing every vertex with two adjacent copies; see Figure~\ref{fig:graph_gprime}.
The number of such perfect matchings can be related to the number of cycle covers of the original graph.
However, the considered perfect matchings can be counted within the claimed complexity
by either employing the previous ``algebraized'' dynamic programming algorithm, or the algorithm based on inclusion-exclusion branching (our presentation chooses the latter).


Applying this approach na\"ively would give us a polynomial space algorithm with running time $8^{\td}\cdot n^{\Oh(1)}$. 
We improve the running time to $5^{\td}\cdot n^{\Oh(1)}$ by employing several observations about the symmetries of recursive calls of our algorithms, 
in a similar way as in the algorithm for $\#k$-\textsc{Multi-Set-Cover} of Nederlof~\cite{mscnederlof}.

\paragraph{Organization of the paper.} The remainder of the paper is devoted to the proof of Theorem~\ref{main-theorem}.
In Section~\ref{sec:prel} we introduce the notation and present basic definitions.
In Section~\ref{sec:cut} we discuss the Cut\&Count technique in a self-contained manner and explain the Cut part.
In Section~\ref{sec:cyctomatch} we reduce the Count part to counting perfect matchings in an auxiliary graph.
In Section~\ref{sec:count} we give an algorithm for counting such matchings, and in Section~\ref{sec:veri} we subsequently verify the correctness of the algorithm.
We conclude with several open questions in Section~\ref{sec:conc}.

\section{Preliminaries}\label{sec:prel}

\paragraph*{Notation.} 
For a graph $G$, by $\cc(G)$ we denote number of connected components of $G$.
Let $F$ be a subset of edges of $G$.
By $\cc(F)$ we denote the number of connected components of the graph consisting of all the edges of $F$ and vertices incident to them.
For a vertex $u$, by $\deg_F(u)$ we mean the number of edges of $F$ incident to $u$.
Then $F$ is a {\em{matching}} if $\deg_F(u)\in \{0,1\}$ for every vertex $u$, is a {\em{perfect matching}} if $\deg_F(u)=1$ for every vertex $u$, 
and is a {\em{partial cycle cover}} if $\deg_F(u)\in \{0,2\}$ for every vertex $u$.
Note that thus we treat partial cycle covers as sets of edges.


A {\em{cut}} of a set $U$ is just an ordered partition of $U$ into two sets, that is, 
a pair $(L,R)$ such that $L\cap R=\emptyset$ and $L\cup R=U$.
A cut $(L,R)$ of the vertex set of a graph is {\em{consistent}} with a subset of edges $F$ if there is no edge in $F$ with one endpoint in $L$ and second in $R$.

For a function $f$ and elements $x,y$, where $x$ is not in the domain of $f$, by $f[x \mapsto y]$ we denote the function obtained from $f$ by extending its domain by $x$ and setting $f(x)=y$.

We use the $\Os(\cdot)$ notation to hide factors polynomial in the input size.
For convenience, throughout the paper we assume the RAM model: every integer takes a unit of space and arithmetic operations on integers have unit cost.
However, it can be easily seen that all the numbers appearing during the computation have bit length bounded polynomially in the input size.
Since we never specify the polynomial factors in the time or space complexity of our algorithms, without any influence on the claimed asymptotic bounds we may
assume that the representation of any number takes polynomial space and arithmetic operations on the numbers take polynomial time.

\paragraph*{Treedepth.} 
A {\em{rooted forest}} is a directed acyclic graph $T$ where every vertex has outdegree at most $1$.
The vertices of outdegree $0$ in $T$ are called the {\em{roots}}.
Whenever a vertex $u$ is reachable from a vertex $v$ by a directed path in $T$, we say that $u$ is an {\em{ancestor}} of $v$, and $v$ is a {\em{descendant}} of $u$.
Note that every vertex is its own ancestor as well as descendant.
The {\em{depth}} of a rooted forest is the maximum number of vertices that can appear on a directed path in it.

We use the following notation from previous works~\cite{stacs2020,pw-stacs2016}.
For a vertex $u$ of a rooted forest $T$, we denote:
\begin{align*}
    \subtree[u] & \coloneqq \{v\colon u \textrm{ is ancestor of } v\}, & \qquad \subtree(u) & \coloneqq \subtree[u]\setminus \{u\}, \\
    \tail[u] & \coloneqq \{v\colon v\textrm{ is ancestor of } u\}, & \qquad \tail(u) & \coloneqq \tail[u]\setminus \{u\}, \\
 \broom[u] & \coloneqq \tail[u]\cup \subtree[u]. & &
\end{align*}
Additionally, $\child(u)$ denotes the set of children of $u$, whereas $\parent(u)$ is the {\em{parent}} of $u$, that is, the only outneighbor of $u$.
If $u$ is a root, we set $\parent(u)=\bot$.

For a graph $G$, an {\em{elimination forest}} of $G$ is a rooted forest $T$ on the same vertex set as $G$ that satisfies the following property:
whenever $uv$ is an edge in $G$, then in $T$ either $u$ is an ancestor of $v$, or $v$ is an ancestor of $u$.
The {\em{treedepth}} of a graph is the minimum possible depth of an elimination forest of $G$.

\paragraph*{Isolation Lemma.}
The only source of randomness in our algorithm is the Isolation Lemma of Mulmuley et al.~\cite{isolation-lemma}.
Suppose $U$ is a finite set and $\omega\colon U\to \Z$ is a weight function on $U$.
We say that $\omega$ {\em{isolates}} a non-empty family of subsets $\Ff\subseteq 2^U$ if there is a unique $S\in \Ff$ such that
$$\omega(S) = \min_{X\in \Ff}\,\omega(X),$$
where $\omega(X)\coloneqq \sum_{x\in X}\omega(x)$. Then the Isolation Lemma can be stated as follows.


\begin{lemma}[Isolation Lemma~\cite{isolation-lemma}]\label{isolation-lemma}
    Let $U$ be a finite set and $\Ff\subseteq 2^U$ be a non-empty family of subsets of $U$.
    Suppose for every $u\in U$ we choose its weight $\omega(u)$ uniformly and independently at random from the set $\{1,\ldots,N\}$, where $N\in \N$.
    Then $\omega$ isolates $\Ff$ with probability at least $1-\frac{|U|}{N}$.
\end{lemma}


\section{The Cut part}\label{sec:cut}

We now proceed to the proof of Theorem~\ref{main-theorem}. Throughout the proof we fix the input graph $G=(V,E)$, its elimination forest $T$ of depth $d$, and numbers $k,\ell\in \N$.
We may assume that $G$ is connected, as otherwise we may apply the algorithm to each connected component separately. Thus $T$ has to be a tree, so we will call it an {\em{elimination tree}} to avoid confusion.
Also, we denote $n\coloneqq |V|$.

As mentioned before, we shall apply the Cut\&Count technique of Cygan et al.~\cite{focs2011}. 
This technique consists of two parts: the Cut part and the Count part.
The idea is that in the first part, we relax the connectivity requirements and show that it is enough to count the number of relaxed solutions together with cuts consistent with them, 
as this number is congruent to the number of non-relaxed solutions modulo a power of $2$. 
The Isolation Lemma is used here to ensure that with high probability, the number of solutions does not accidentally cancel out modulo this power of $2$.
More precisely, having drawn a weight function at random, for each possible total weight $w$ we count the number of solutions of total weight $w$.
Then the Isolation Lemma asserts that, with high probability, for some $w$ there will be a unique solution of total weight $w$.
Then comes the Count part, where the goal is to efficiently count the number of relaxed solutions together with cuts consistent with them.

We refer the reader to~\cite{focs2011} for a more elaborate discussion of the Cut\&Count technique, while now we apply it to the particular case of \PCC.
A relaxed solution is just a partial cycle cover consisting of $\ell$ edges. Then a solution is a relaxed solution that spans at most $k$ cycles.
Formally, the sets of \emph{solutions} ($\Ss$) and \emph{relaxed solutions} ($\Rr$) are defined as follows:
\begin{align*}
\Rr\coloneqq &\ \{\,F\subseteq E\ \colon\ |F|=\ell\textrm{ and }\deg_F(u)\in \{0,2\}\textrm{ for every }u\in V\,\};\\
\Ss\coloneqq &\ \{\,F\in \Rr\ \colon\ \cc(F)\leq k\,\}.
\end{align*}
Suppose now that the input graph $G$ is supplied with a weight function on edges $\omega\colon E\to \Z$. Then we can stratify the families above using the total weight.
That is, for every $w\in \Z$ we define:
$$\Rr_w\coloneqq \{\,F\in \Rr\ \colon\ \omega(F)=w\}\qquad\textrm{and}\qquad \Ss_w\coloneqq \{\,F\in \Ss\ \colon\ \omega(F)=w\,\}.$$
Now, let
$$\Cc_w\coloneqq \{\,(F,(L,R))\ \colon\ F\in \Rr_w\textrm{ and }(L,R)\textrm{ is a cut of }V\textrm{ consistent with }F\,\}.$$
The following observation is the key idea in the Cut\&Count technique.

\begin{lemma}\label{lem:cut}
 For every $w\in \Z$, we have
 $$|\Cc_w|\equiv \sum_{F\in \Ss_w} 2^{n-\ell+\cc(F)}\mod 2^{n-\ell+k+1}.$$
\end{lemma}
\begin{proof}
 Observe that for each $F\in \Rr_w$ there are exactly $2^{n-\ell+\cc(F)}$ cuts of $V$ consistent with it, because each of the $\cc(F)$ cycles spanned by $F$ can be on either side of the cut, and similarly each of
 $n-\ell$ vertices not incident to the edges of $F$ can be on either side of the cut. Hence $|\Cc_w|=\sum_{F\in \Rr_w} 2^{n-\ell +\cc(F)}$.
 However, for every $F\in \Rr_w\setminus \Ss_w$ the term $2^{n-\ell+\cc(F)}$ is divisible by $2^{n-\ell+k+1}$ since $\cc(F) \geq k+1$, and thus $\sum_{F\in \Rr_w} 2^{\cc(F)}\equiv \sum_{F\in \Ss_w} 2^{\cc(F)}\bmod 2^{n-\ell+k+1}$.
\end{proof}

In the next sections we will present the Count part of the technique, which boils down to proving the following lemma.

\begin{lemma}\label{lem:count}
  Given $w\in \Z$ and a weight function $\omega\colon E\to \{1,\ldots,N\}$, where $N=\Os(1)$, the number $|\Cc_w|$ can be computed in time $\Os(5^d)$ and space $\Os(1)$.
\end{lemma}

We now show how to combine Lemma~\ref{lem:cut} with Lemma~\ref{lem:count} to prove Theorem~\ref{main-theorem}.

\begin{proof}[Proof of Theorem~\ref{main-theorem} assuming Lemma~\ref{lem:count}]
 Let $N=2|E|$. The algorithm proceeds as follows. 
 First, for every edge $e\in E$, sample its weight $\omega(e)$ uniformly and independently at random from the set $\{1,\ldots,N\}$.
 Next, for each $w\in \{1,\ldots,N|E|\}$ compute the number $|\Cc_w|$ in time $\Os(5^d)$ and space $\Os(1)$ using the algorithm of Lemma~\ref{lem:count}.
 If for some $w$ the number $|\Cc_w|$ is not divisible by $2^{n-\ell+k+1}$, then output that there exists a solution.
 Otherwise, output that there is no solution.
 
 It is clear that the algorithm runs in time $\Os(5^d)$ and uses $\Os(1)$ space, so it remains to argue the correctness.
 On one hand, observe that if $\Ss=\emptyset$, then $\Ss_w=\emptyset$ for all $w\in \Z$, hence by Lemma~\ref{lem:cut} all the computed numbers $|\Cc_w|$ will be indeed divisible by $2^{n-\ell+k+1}$.
 Therefore, there are no false positives.
 On the other hand, if $\Ss\neq \emptyset$, then the Isolation Lemma implies that with probability at least $\frac{|E|}{N}=\frac{1}{2}$ there exists $w\in \Z$ such that $|\Ss_w|=1$.
 Note that it must hold that $w\in \{1,\ldots,N|E|\}$.
 Denoting $\Ss_w=\{F\}$, by Lemma~\ref{lem:cut} we have $|\Cc_w|\equiv 2^{n-\ell+\cc(F)}\bmod 2^{n-\ell+k+1}$. As $\cc(F)\leq k$, the number $|\Cc_w|$ is then not divisible by $2^{n-\ell+k+1}$ and the algorithm 
 correctly reports the positive outcome.
\end{proof}

Hence, it remains to prove Lemma~\ref{lem:count}.

\section{From cycle covers to matchings}\label{sec:cyctomatch}

For the proof of Lemma~\ref{lem:count}, instead of counting the number of suitable partial cycle covers, 
we find it more convenient to count the number of perfect matchings in an auxiliary graph.
Note, that this concept is natural when using \emph{inclusion-exclusion branching}
technique. A similar auxiliary graph arises in the algorithm for 
$\#k$-\textsc{Multi-Set-Cover}~\cite{mscnederlof}.

We define a graph $G'$ as follows. The vertex set $V'$ of $G'$ is
$$V'\coloneqq \{\,u^0,u^1\ \colon\ u\in V\,\}.$$
That is, we put two copies of each vertex of $G$ into the vertex set of $G'$.
The edge set $E'$ of $G'$ is the union of the following two sets:
\begin{align*}
 E'_0 & \coloneqq \{\,u^0u^1\ \colon\ u\in V\,\},\\
 E'_1 & \coloneqq \{\,u^0v^0,u^0v^1,u^1v^0,u^1v^1\ \colon\ uv\in E\,\}.
\end{align*}
In other words, for every vertex $u\in V$ we put an edge in $E'_0$ connecting the two copies of $u$ in $V'$,
while for every edge $uv\in E$ we put four different edges in $E'_1$, each
connecting a copy of $u$ with a copy of $v$ in $V'$. See Figure~\ref{fig:graph_gprime} for a visualization of the construction of $G'$.

Let $\pi\colon E'_1\to E$ be the natural projection from $E'_1$ to $E$: for each $uv\in E$ and $s,t\in \{0,1\}$, we set $\pi(u^sv^t)=uv$.
We extend the mapping $\pi$ to all subsets $F\subseteq E'$ by setting $\pi(F)\coloneqq \pi(F\cap E'_1)$.
We also extend the weight function $\omega$ to the edges of $E'$ by putting $\omega(e)=0$ for each $e\in E'_0$ and $\omega(e)=\omega(\pi(e))$ for each $e\in E'_1$.

A set of edges $F$ in $G'$ shall be called {\em{simple}} if for every $e\in E$, we have
$$|F\cap \pi^{-1}(e)|\leq 1.$$
For now, we mainly focus on simple perfect matchings in $G'$. We observe that they are in correspondence with partial cycle covers in $G$, as explained next.

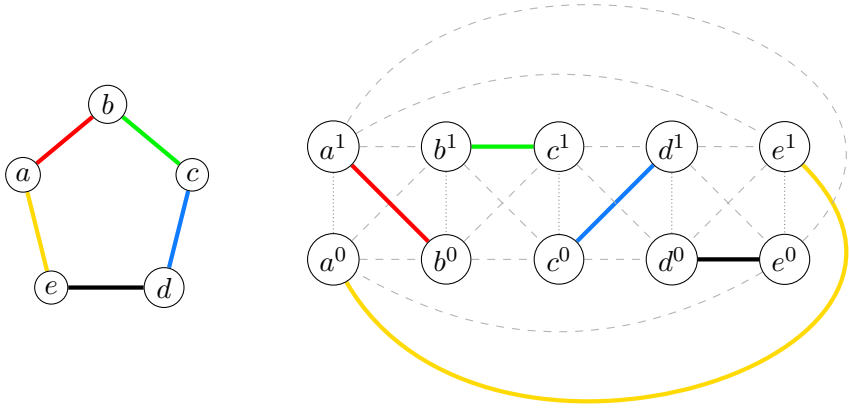
\begin{figure}[ht!]
    \centering
    \begin{tikzpicture}[scale=1.5]
	\begin{pgfonlayer}{nodelayer}
		\node [style=vertex] (0) at (-12.5, 4.5) {$a$};
		\node [style=vertex] (1) at (-10, 2.5) {$d$};
		\node [style=vertex] (2) at (-11, 5.75) {$b$};
		\node [style=vertex] (3) at (-9.5, 4.5) {$c$};
		\node [style=vertex] (4) at (-7, 3) {$a^0$};
		\node [style=vertex] (5) at (-5, 3) {$b^0$};
		\node [style=vertex] (6) at (-3, 3) {$c^0$};
		\node [style=vertex] (7) at (-1, 3) {$d^0$};
		\node [style=vertex] (8) at (-7, 5) {$a^1$};
		\node [style=vertex] (9) at (-5, 5) {$b^1$};
		\node [style=vertex] (10) at (-3, 5) {$c^1$};
		\node [style=vertex] (11) at (-1, 5) {$d^1$};
		\node [style=none] (12) at (-12, 2.5) {};
		\node [style=vertex] (13) at (-12, 2.5) {};
		\node [style=vertex] (14) at (-12, 2.5) {$e$};
		\node [style=vertex] (15) at (1, 5) {$e^1$};
		\node [style=vertex] (16) at (1, 3) {$e^0$};
	\end{pgfonlayer}
	\begin{pgfonlayer}{edgelayer}
		\draw [style=blue] (1) to (3);
		\draw [style=green] (3) to (2);
		\draw [style=red] (2) to (0);
		\draw [style=gray, in=180, out=0] (4) to (5);
		\draw [style=gray] (6) to (7);
		\draw [style=green] (9) to (10);
		\draw [style=gray] (5) to (6);
		\draw [style=gray] (10) to (11);
		\draw [style=gray] (8) to (9);
		\draw [style=gray] (4) to (9);
		\draw [style=gray] (9) to (6);
		\draw [style=blue] (6) to (11);
		\draw [style=red] (8) to (5);
		\draw [style=gray] (5) to (10);
		\draw [style=gray] (10) to (7);
		\draw [style=dotted] (7) to (11);
		\draw [style=dotted] (6) to (10);
		\draw [style=dotted] (9) to (5);
		\draw [style=dotted] (4) to (8);
		\draw [style=new edge style 1] (0) to (14);
		\draw [style=new edge style 0] (14) to (1);
		\draw [style=gray] (11) to (15);
		\draw [style=gray] (15) to (7);
		\draw [style=gray] (11) to (16);
		\draw [style=new edge style 0] (16) to (7);
		\draw [style=dotted] (15) to (16);
		\draw [style=gray, bend left] (8) to (15);
		\draw [style=gray, bend left] (16) to (4);
		\draw [style=new edge style 1, in=-60, out=-45, looseness=1.50] (15) to (4);
		\draw [style=gray, in=60, out=45, looseness=1.50] (16) to (8);
	\end{pgfonlayer}
\end{tikzpicture}
    \vspace{-1cm}
	\caption{Construction of the graph $G'$ from $G=C_5$, together with a simple perfect matching $M$ that projects in $\pi$ to $E$.
	Solid edges represent $M$, dashed edges belong to $E'_1\setminus M$, while dotted edges comprise $E'_0$.}
	\label{fig:graph_gprime}
\end{figure}

\begin{lemma}\label{lem:covers-matchings}
 For every simple perfect matching $M$ in $G'$, the set $\pi(M)$ is a partial cycle cover in $G$ of size $|M\cap E'_1|$.
 Moreover, for every partial cycle cover $F$ in $G$, there are exactly $2^{|F|}$ simple perfect matchings $M$ in $G'$ for which $F=\pi(M)$.
\end{lemma}
\begin{proof}
For the first assertion, observe that if for some $u\in V$, $M$ matches $u^0$ with some vertex of the form $v^t$ for $v\neq u$,
then $u^{1}$ has to be matched by $M$ with some vertex of the form $\bar{v}^{\bar{t}}$, where $\bar{v}\neq u$ and $\bar{v}\neq v$; the latter inequality follows from the simplicity of $M$.
Then $\deg_{\pi(M)}(u)=2$.
On the other hand, if $M$ matches $u^0$ with $u^1$, then $\deg_{\pi(M)}(u)=0$.
Thus $\deg_{\pi(M)}(u)\in \{0,2\}$ for every vertex $u\in V$, so $\pi(M)$ is a partial cycle cover in $G$.
The fact that $|\pi(M)|=|M\cap E'_1|$ follows directly from the construction and the simplicity of $M$.

For the second assertion, consider any partial cycle cover $F$ in $G$. Let $U$ be the set of vertices incident to the edges of $F$; then $|U|=|F|$.
Let a {\em{binding}} be any function $f\colon U\to F$ such that for every $u\in U$, $f(u)$ is one of the two edges of $F$ incident to $u$. 
Observe that if for a binding $f$ we define
$$M(f) = \{v^0v^1\colon v\in V\setminus U\} \cup \{u^{[f(u)=uv]}v^{[f(v)=uv]}\colon uv\in F\},$$
where $[\varphi]=1$ if condition $\varphi$ holds and $[\varphi]=0$ otherwise, then $M(f)$ is a simple perfect matching in $G'$ satisfying $\pi(M(f))=F$.
Clearly, matchings $M(f)$ obtained for different bindings $f$ are pairwise different. 
Moreover, it is easy to see that every simple perfect matching $M$ satisfying $\pi(M)=F$ is of the form $M=M(f)$ for some binding $f$. 
Indeed, for every $u\in U$ we just set $f(u)=\pi(e)$, where $e\in E'_1$ is the edge of $M$ that is incident to $u^1$.
Since the total number of different bindings is $2^{|F|}$, we conclude that there are exactly $2^{|F|}$ simple perfect matchings $M$ in $G'$ satisfying $\pi(M)=F$.
\end{proof}

Lemma~\ref{lem:covers-matchings} motivates introducing the following analogues of the sets $\Cc_w$. For $w\in \Z$, we define
\begin{align*}
\Mm_w & \coloneqq \{\,(M,(L,R)) & \colon & M \textrm{ is a simple perfect matching in }G',\ \omega(M)=w,\\
& & & |M\cap E'_1|=\ell, \textrm{ and } (L,R)\textrm{ is a cut of }V\textrm{ consistent with }\pi(M)\,\}. 
\end{align*}

Since for every simple perfect matching $M$ in $G'$ we have $\omega(M)=\omega(\pi(M))$, from Lemma~\ref{lem:covers-matchings} we immediately obtain the following.

\begin{corollary}
 For every $w\in \Z$, we have $|\Mm_w|=2^\ell\cdot |\Cc_w|$.
\end{corollary}

Therefore, to prove Lemma~\ref{lem:count} it suffices to apply the algorithm provided by the following lemma and divide the outcome by $2^\ell$.

\begin{lemma}\label{lem:count2}
  Given $w\in \Z$ and a weight function $\omega\colon E\to \{1,\ldots,N\}$, where $N=\Os(1)$, the number $|\Mm_w|$ can be computed in time $\Os(5^d)$ and space $\Os(1)$.
\end{lemma}

We are left with proving Lemma~\ref{lem:count2}.

\section{The Count part}\label{sec:count}

In this section we execute the Count part of the technique by proving Lemma~\ref{lem:count2}. Let us first discuss the intuition behind the approach.

The basic idea is that we will compute the number $|\Mm_w|$ using bottom-up dynamic programming over the given elimination tree $T$.
In order to achieve polynomial space complexity, this dynamic programming will be cast as a standard recursion, but
for this to work, we need that the recurrence equations governing the dynamic programming have a specific form.
In essence, whenever we compute an entry of the dynamic programming table at some vertex $u$, the value should be obtained as a simple aggregation of single entries from the tables of the children of $u$.
The most straightforward approach to computing $|\Mm_w|$ would be to count partial perfect matchings and to remember, in the states corresponding to $u$, subsets of $\tail[u]$ consisting of vertices matched to $\subtree(u)$.
This would yield a dynamic programming algorithm that is {\em{not}} of the form required for the space complexity reduction.
However, we show that by counting different objects than partial perfect matchings, and using the inclusion-exclusion principle at every computation step, we can reorganize the computation so that the space reduction is possible.

We remark that even though at the end of the day our algorithm relies only on basic ideas such as branching and inclusion-exclusion,
there is a deeper intuition behind the definitions of the computed values.
In fact, from the right angle our algorithm can be seen as an application of the technique of {\em{saving space by algebraization}}, introduced by Lokshtanov and Nederlof~\cite{algebraization},
which boils down to applying the Fourier transform on the lattice of subsets in order to turn subset convolutions into pointwise products.
We refer the reader to~\cite{stacs2020,furer-1,furer-2,pw-stacs2016} for other applications of this technique in the context of treedepth-based algorithms.



\newcommand{\Fun}{\mathsf{Func}}
\newcommand{\lf}{\mathsf{left}}
\newcommand{\sheaf}{\mathsf{sheaf}}
\newcommand{\zero}{0}
\newcommand{\onel}{1_{\mathsf{L}}}
\newcommand{\oner}{1_{\mathsf{R}}}
\newcommand{\twol}{2_{\mathsf{L}}}
\newcommand{\twor}{2_{\mathsf{R}}}

\paragraph*{Partial objects.} We start with defining the partial objects that will be counted by the algorithm.

For every vertex $u\in V$, let us order the children of $u$ in an arbitrary manner. Thus, every non-leaf vertex $u$ has a unique left-most (first in the order) child. 
For every $u\in V$, let $\lf(u)$ be the left-most leaf descendant of $u$, that is, the leaf obtained by starting at $u$ and iteratively moving to the left-most child of the current vertex until a leaf is found.
This induces a mapping $\ol{\lf}(\cdot)$ on the edges of $E'$ as follows: 
for an edge $e=u^0u^1\in E'_0$, we put $\ol{\lf}(e)\coloneqq\lf(u)$, while for
an edge of the form $e=u^sv^t\in E'_1$, where $u\neq v$ and $s,t\in \{0,1\}$, we
let $\ol{\lf}(e)\coloneqq \lf(v)$, 
where $v$ is the descendant of $u$ in $T$. Now, for every $u\in V$ we define the {\em{sheaf}} of $u$ as follows:
$$\sheaf[u]\coloneqq \bigcup_{v \text{ is a leaf in } \subtree[u]} \ol{\lf}^{-1}(v)\ \subseteq\ E'.$$

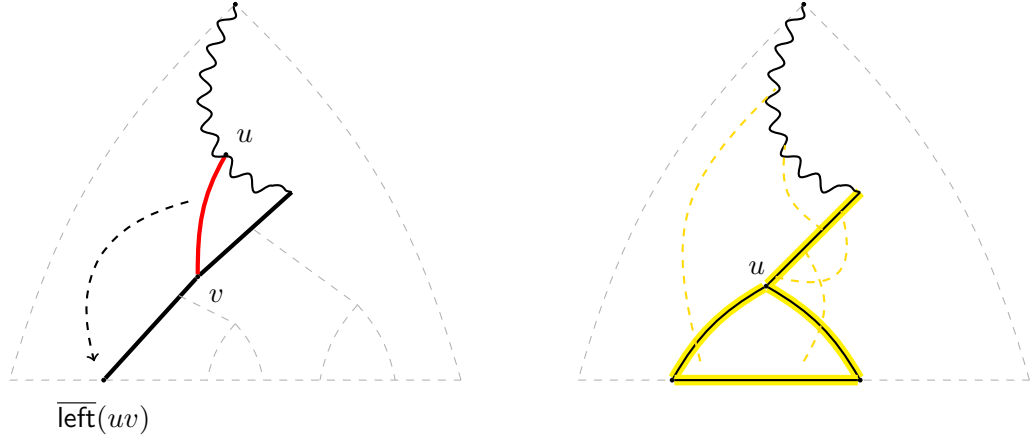
\begin{figure}[ht!]
    \centering
    \begin{subfigure}[b]{0.45\textwidth}
		\begin{tikzpicture}
	\begin{pgfonlayer}{nodelayer}
		\node [style=new style 0] (10) at (0, 8) {};
		\node [style=none] (11) at (-6, -2) {};
		\node [style=none] (12) at (6, -2) {};
		\node [style=none] (14) at (1.5, 3) {};
		\node [style=new style 0] (15) at (-1, 0.75) {};
		\node [style=new style 0] (16) at (-3.5, -2) {};
		\node [style=new style 0] (17) at (-0.25, 4) {};
		\node [style=none] (18) at (0.25, 4.5) {$u$};
		\node [style=none] (20) at (-0.5, 0.25) {$v$};
		\node [style=none] (21) at (-1.25, 2.75) {};
		\node [style=none] (22) at (-3.75, -1.5) {};
		\node [style=none] (24) at (-3.5, -3) {$\overline{\mathsf{left}}(uv)$};
		\node [style=none] (26) at (4.25, -2) {};
		\node [style=none] (27) at (0.5, 2) {};
		\node [style=none] (28) at (3.25, 0) {};
		\node [style=none] (29) at (-1.5, 0.25) {};
		\node [style=none] (30) at (0, -0.5) {};
		\node [style=none] (33) at (2.25, -2) {};
		\node [style=none] (34) at (-0.75, -2) {};
		\node [style=none] (35) at (0.75, -2) {};
	\end{pgfonlayer}
	\begin{pgfonlayer}{edgelayer}
		\draw [style=gray, bend right=15] (10) to (11.center);
		\draw [style=gray] (11.center) to (12.center);
		\draw [style=gray, bend right=15] (12.center) to (10);
		\draw [style=new edge style 0] (14.center) to (15);
		\draw [style=new edge style 0] (15) to (16);
		\draw [style=red, in=-120, out=90] (15) to (17);
		\draw [style=arrow, in=105, out=-165, looseness=1.25] (21.center) to (22.center);
		\draw [style=gray] (27.center) to (28.center);
		\draw [style=gray] (29.center) to (30.center);
		\draw [style=gray, bend left=15] (28.center) to (26.center);
		\draw [style=gray, bend right=15] (28.center) to (33.center);
		\draw [style=gray, bend right=15] (30.center) to (34.center);
		\draw [style=gray, bend left=15] (30.center) to (35.center);
		\draw [style=wave, in=165, out=-120, looseness=1.25] (10) to (14.center);
	\end{pgfonlayer}
\end{tikzpicture}
    \end{subfigure}
    \begin{subfigure}[b]{0.45\textwidth}
    	\begin{tikzpicture}
	\begin{pgfonlayer}{nodelayer}
		\node [style=new style 0] (10) at (0, 8) {};
		\node [style=none] (11) at (-6, -2) {};
		\node [style=none] (12) at (6, -2) {};
		\node [style=none] (14) at (1.5, 3) {};
		\node [style=new style 0] (15) at (-1, 0.5) {};
		\node [style=new style 0] (16) at (-3.5, -2) {};
		\node [style=new style 0] (17) at (1.5, -2) {};
		\node [style=none] (18) at (-2.75, -1.5) {};
		\node [style=none] (19) at (-0.75, 5.75) {};
		\node [style=none] (20) at (0, -1.5) {};
		\node [style=none] (21) at (0, 1.5) {};
		\node [style=none] (22) at (0.5, 2) {};
		\node [style=none] (23) at (-0.5, 4.25) {};
		\node [style=none] (24) at (-0.75, 0.75) {};
		\node [style=none] (26) at (1, 2.5) {};
		\node [style=none] (27) at (-1.25, 1) {$u$};
	\end{pgfonlayer}
	\begin{pgfonlayer}{edgelayer}
		\draw [style=gray, bend right=15] (10) to (11.center);
		\draw [style=gray] (11.center) to (12.center);
		\draw [style=gray, bend right=15] (12.center) to (10);
		\draw [style=sheaf] (14.center) to (15);
		\draw [style=sheaf, bend right=15] (15) to (16);
		\draw [style=wave, in=165, out=-120, looseness=1.25] (10) to (14.center);
		\draw [style=sheaf, bend left=15] (15) to (17);
		\draw [style=sheaf] (17) to (16);
		\draw [style=golden, bend left, looseness=1.25] (18.center) to (19.center);
		\draw [style=golden, bend right, looseness=1.25] (20.center) to (21.center);
		\draw [style=golden, bend left, looseness=1.25] (22.center) to (23.center);
		\draw [style=golden, bend right=60, looseness=1.75] (24.center) to (26.center);
	\end{pgfonlayer}
\end{tikzpicture}
    \end{subfigure}
    \caption{Schematic definitions of $\ol{\mathsf{left}}$ and $\mathsf{sheaf}$. 
    Left panel presents the definition of the mapping $\ol{\mathsf{left}}(uv)$. Since $v$ is a descendant of $u$, $\ol{\mathsf{left}}(uv)$ is the leftmost descendant of $v$.  
    Right panel presents the definition of $\mathsf{sheaf}[u]$. Any edge with the lower endpoint in the yellow highlighted region is a part of $\mathsf{sheaf}[u]$.}
\label{fig:definitions_sheaf}
\end{figure}

See Figure~\ref{fig:definitions_sheaf} for a schematic presentation of
$\sheaf[u]$. The next observation follows immediately from the definition.

\begin{observation}\label{lem:sheaf-partition}
 For every vertex $u$ that is not a leaf in $T$, the family $\{\sheaf[v]\colon v\in \child(u)\}$ is a partition of $\sheaf[u]$.
\end{observation}

We now move to the description of partial objects.
In the following, we will use the convention that if $Z\subseteq V$, then we write $Z'\coloneqq \{z^0,z^1\colon z\in Z\}\subseteq V'$ for the set of copies of vertices of $Z$ in $G'$.
We first introduce the following notion that facilitates the definition of the partial objects:
\begin{definition}\label{def:compatible}
Suppose $X$ and $Y$ are two disjoint subsets of $V$. Further, suppose $S\subseteq E'$ is a set of edges whose all endpoints are in $X'\cup Y'$.
For a function $f\colon Y\to \{\zero,\onel,\oner,\twol,\twor\}$, we shall say that a pair $(F,(L,R))$, where $F\subseteq S$ and $(L,R)$ is a cut of $X\cup Y$,
is {\em{compatible}} with $f$ if the following properties hold:
\begin{itemize}
 \item $F$ is simple and consistent with the cut $(L',R')$;
 \item $f^{-1}(\{\onel,\twol\})\subseteq L$ and $f^{-1}(\{\oner,\twor\})\subseteq R$;
 \item for every $y\in Y$ with $f(y)=\zero$, no edge of $F$ is adjacent to $y^0$ or to $y^1$;
 \item for every $y\in Y$ with $f(y)\in \{\onel,\oner\}$, no edge of $F$ is adjacent to $y^1$; and
 \item every vertex of $X'$ is incident to some edge of $F$.
\end{itemize}
\end{definition}

In essence, to define partial objects constructed for a vertex $u$, in the definition above we take $Y$ to be the tail of $u$, $X$ to be the subtree of $u$, and $S$ to be the sheaf of $u$.
Then with every function $f\colon Y\to \{\zero,\onel,\oner,\twol,\twor\}$ we can associate all partial objects that are compatible with it.
However, it makes a difference whether we include or exclude the vertex $u$ from $Y$; that is, whether we consider $Y=\tail[u]$ and $X=\subtree(u)$, or $Y=\tail(u)$ and $X=\subtree[u]$.
Therefore, we distinguish \emph{inclusive} and \emph{exclusive} partial objects.
\begin{definition}
For every $u\in V$ and every function $f\colon \tail[u]\to \{\zero,\onel,\oner,\twol,\twor\}$, 
we define the set of {\em{inclusive partial objects}} for $u$ and $f$, denoted $\Mm[u,f]$, as the set of all pairs $(F,(L,R))$ that are compatible with $f$, where $X=\subtree(u)$, $Y=\tail[u]$, and $S=\sheaf[u]$.
\end{definition}
\begin{definition}
For every $u \in V$ and for every function $f\colon \tail(u)\to \{\zero,\onel,\oner,\twol,\twor\}$, 
the set of {\em{exclusive partial objects}} for $u$ and $f$, denoted $\Mm(u,f)$, is the set of all pairs $(F,(L,R))$ that are compatible with $f$, where $X=\subtree[u]$, $Y=\tail(u)$, and $S=\sheaf[u]$.
\end{definition}
Observe that both in the inclusive and in the exclusive case we have $X\cup Y=\broom[u]$, so $(L,R)$ is a cut of $\broom[u]$, and $F\subseteq S=\sheaf[u]$.
Note also that we {\em{do not}} require $F$ to be a matching in~$G'$.
For convenience, by $\Fun[u]$ and $\Fun(u)$ we shall denote the sets of all functions from $\tail[u]$, respectively $\tail(u)$, to $\{\zero,\onel,\oner,\twol,\twor\}$.

Finally, we stratify the sets $\Mm[u,f]$ by defining, for all $a,b,c\in \N$, $\Mm_{a,b,c}[u,f]$ as the set of all the pairs $(F,(L,R))\in \Mm[u,f]$ such that $\omega(F)=a$, $|F|=b$, and $|F\cap E'_1|=c$.
Sets $\Mm_{a,b,c}(u,f)$ are defined analogously. The following lemma follows easily from the definitions.

\begin{lemma}\label{lem:root}
 If $r$ is the root of $T$, then $\Mm_w = \Mm_{w,n,\ell}(r,\emptyset)$.
\end{lemma}
\begin{proof}
 We observe that if $F\subseteq E'$ is such that $|F|=n$ and every vertex of $V'$ is incident to at least one edge of $F$, then $F$ has to be a perfect matching in $G'$.
 Then the remaining requirements expressed in the definition of $\Mm_{w,n,\ell}(r,\emptyset)$ exactly correspond to the restrictions on matchings considered in the definition of $\Mm_w$.
\end{proof}

Thus, our goal is to compute all the cardinalities of the sets $\Mm_{a,b,c}[u,f]$ and $\Mm_{a,b,c}(u,f)$, for all relevant choices of $a,b,c,f,u$. 

\paragraph*{Encoding accumulators in formal variables.}
Similarly as in~\cite{stacs2020,pw-stacs2016}, we encode the different choices of $a,b,c\in \N$ as degrees of formal variables $\alpha,\beta,\gamma$, so that all the relevant values $|\Mm_{a,b,c}[u,f]|$ can be stored
as coefficients of one polynomial from $\Z[\alpha,\beta,\gamma]$, and similarly for the values $|\Mm_{a,b,c}(u,f)|$.
Formally, for each $u\in V$ and each $f\in \Fun[u]$ we define the polynomial $P[u,f]\in \Z[\alpha,\beta,\gamma]$ as
$$P[u,f] = \sum_{a,b,c\in \N} |\Mm_{a,b,c}[u,f]|\cdot \alpha^a\beta^b\gamma^c,$$
and for each $f\in \Fun(u)$ we define the polynomial $P(u,f)\in \Z[\alpha,\beta,\gamma]$ as
$$P(u,f) = \sum_{a,b,c\in \N} |\Mm_{a,b,c}(u,f)|\cdot \alpha^a\beta^b\gamma^c.$$
Observe that since we assume that the weight function $\omega$ only assigns weights in $\{1,\ldots,N\}$, in the formula above the numbers $|\Mm_{a,b,c}[u,f]|$ and $|\Mm_{a,b,c}(u,f)|$ can be non-zero only
for $a\leq N\cdot |E|$, $b\leq |E|$, and $c\leq |E|$. Thus, $P[u,f]$ and $P(u,f)$ are indeed polynomials.

As argued above, each polynomial $P[u,f]$ and $P(u,f)$ has total degree at most $(N+2)\cdot |E|$, hence it is a sum of a polynomial (in $n$) number of monomials. 
Hence, we may represent each of the polynomials $P[u,f]$ and  $P(u,f)$ by just storing a polynomial-size table of the coefficients of the monomials. 
Thus, the representation of each polynomial $P[u,f]$ and $P(u,f)$ takes polynomial space, and arithmetic operations on them can be performed in polynomial time.

\paragraph*{The computation.}
The idea now is that each polynomial $P(u,f)$ can be computed using polynomials $P[u,f']$ for different extensions $f'$ of $f$,
while each polynomial $P[u,f]$ for a non-leaf $u$ can be computed using polynomials $P(v,f)$ for $v$ ranging over the children of $u$.
Moreover, polynomials $P[u,f]$ can be computed in polynomial time whenever $u$ is a leaf. These statements are encapsulated in the following lemmas, whose proofs are postponed to Section~\ref{sec:veri}.

\begin{lemma}\label{lem:exclusive-step}
 For every $u\in V$ and $f\in \Fun(u)$, we have
 $$P(u,f) = P[u,f[u\mapsto \twol]]+P[u,f[u\mapsto \twor]]-2\cdot P[u,f[u \mapsto \onel]]-2\cdot P[u,f[u \mapsto \oner]]+P[u,f[u \mapsto \zero]].$$
\end{lemma}

\begin{lemma}\label{lem:inclusive-step}
 For every $u\in V$ which is not a leaf in $T$, and every $f\in \Fun[u]$, we have
 $$P[u,f]=\prod_{v\in \child(u)} P(v,f).$$
\end{lemma}

\begin{lemma}\label{lem:leaf-step}
 For every $u\in V$ which is a leaf in $T$, and every $f\in \Fun[u]$, the polynomial $P[u,f]$ can be computed in polynomial time.
\end{lemma}

The proofs of Lemmas~\ref{lem:exclusive-step},~\ref{lem:inclusive-step}, and~\ref{lem:leaf-step} are given in the next section,
but they immediately suggest a recursive method for the computation of polynomials $P[u,f]$ and $P(u,f)$.
We now show that this idea can be used to finish the proof of Lemma~\ref{lem:count2}.

\newcommand{\inclAlgo}{\mathsf{computeInclusive}}
\newcommand{\exclAlgo}{\mathsf{computeExclusive}}

\begin{algorithm}
\SetAlgoLined
\SetKwInOut{Input}{Input}
\SetKwInOut{Output}{Output}
\Input{vertex $u\in V$, function $f\in \Fun(u)$}
\Output{$P(u,f)$}
    \ForEach{$s \in \{\zero,\onel,\oner,\twol,\twor\}$}{
        $P_s := \inclAlgo(u,f[u \mapsto s])$
    }
    return $P_{\twol} + P_{\twol} - 2\cdot P_{\onel}- 2\cdot P_{\oner} + P_{\zero}$
\caption{Procedure $\exclAlgo(u,f)$.}
\label{alg:p_exc}
\end{algorithm}
\begin{algorithm}
\SetAlgoLined
\SetKwInOut{Input}{Input}
\SetKwInOut{Output}{Output}
\Input{vertex $u\in V$, function $f\in \Fun[u]$}
\Output{$P[u,f]$}
    \eIf{$u$\textrm{ is a leaf}}{
     Compute $P$ using Lemma~\ref{lem:leaf-step} for $u$ and $f$
    }{
      $P \coloneqq 1$\\
      \ForEach{$v \in \child(u)$}{
	  $P \coloneqq P \cdot \exclAlgo(v,f)$
      }
    }
    \KwRet $P$
\caption{Procedure $\inclAlgo(u,f)$.}
\label{alg:p_inc}
\end{algorithm}

\begin{proof}[Proof of Lemma~\ref{lem:count2} assuming Lemmas~\ref{lem:exclusive-step},~\ref{lem:inclusive-step},~\ref{lem:leaf-step}]
 We give two mutually recursive procedures---$\exclAlgo(u,f)$ and $\inclAlgo(u,f)$---that compute polynomials $P(u,f)$ and $P[u,f]$, respectively.
 These procedures are presented above using pseudocode as Algorithms~\ref{alg:p_exc} and~\ref{alg:p_inc}.
 In summary, to compute $P(u,f)$ we recursively compute the values $P[u,f[u\mapsto s]]$ for all $s\in \{\zero,\onel,\oner,\twol,\twor\}$, using procedure $\inclAlgo$, 
 and then apply the formula provided by Lemma~\ref{lem:exclusive-step}.
 To compute $P[u,f]$ we either use the base case provided by Lemma~\ref{lem:leaf-step} when $u$ is a leaf, or otherwise we recursively compute the values $P(v,f)$ for all $v\in \child(u)$, using procedure $\exclAlgo$, 
 and multiply them.
 Lemmas~\ref{lem:exclusive-step},~\ref{lem:inclusive-step}, and~\ref{lem:leaf-step} assert that the presented procedures correctly compute the polynomials $P(u,f)$ and $P[u,f]$, for all $u\in V$ and relevant functions $f$.
 Hence, by Lemma~\ref{lem:root}, in order to compute $|\Mm_w|$ it suffices to run the procedure $\exclAlgo(r,\emptyset)$, where $r$ is the root of $T$, and return the coefficient in the obtained polynomial that stands by the
 monomial $\alpha^w\beta^n\gamma^\ell$.
 Thus, it is clear that the presented algorithm is correct.
 
 It remains to argue that the algorithm runs in time $\Os(5^d)$ and uses polynomial space.
 As for the space complexity, observe that at each point, the algorithm maintains a recursion stack of depth at most $d$ and the internal data of each recursive call on the stack take polynomial space.
 As for the time complexity, the crucial observation is that throughout the recursion, for every pair $(u,f)$ where $u\in V$ and $f\in \Fun(u)$ we call the procedure $\exclAlgo(u,f)$ exactly once: 
 within the procedure $\inclAlgo(u',f)$ where $u'$ is the parent of $u$, or at the very beginning if $u=r$.
 Similarly, for every pair $(u,f)$ where $u\in V$ and $f\in \Fun[u]$ we call the procedure $\inclAlgo(u,f)$ also exactly once: within the procedure $\exclAlgo(u,f')$, where $f'$ is the restriction of $f$ to $\tail(u)$.
 Thus, the total number of recursive calls executed throughout the algorithm is bounded by the number of pairs $(u,f)$ as above, which is at most $2n\cdot 5^d$.
 As internal operations within each recursive call take polynomial time, we conclude that the total time complexity is $\Os(5^d)$.
\end{proof}

\section{Verification of recursive formulas}\label{sec:veri}


In this section we provide the proofs of Lemmas~\ref{lem:exclusive-step}, Lemmas~\ref{lem:inclusive-step}, and Lemmas~\ref{lem:leaf-step}.

\newcommand{\oMm}{\mathcal{N}}
\newcommand{\oMmL}{\oMm^{\mathsf{L}}}
\newcommand{\oMmR}{\oMm^{\mathsf{R}}}

\begin{proof}[Proof of Lemma~\ref{lem:exclusive-step}]
    Notice that since the postulated equality involves only sums of polynomials,
    it suffices to verify the equality of coefficients by each monomial $\alpha^a\beta^b\gamma^c$.
    In other words, we need to prove that for all $a,b,c\in \N$ we have
    \begin{align}
    |\oMm(u,f)| = &\ |\oMm[u,f[u\mapsto \twol]]+|\oMm[u,f[u\mapsto \twor]]|\nonumber\\
                  &\ -2|\oMm[u,f[u \mapsto \onel]]|-2|\oMm[u,f[u \mapsto \oner]]|+|\oMm[u,f[u \mapsto \zero]]|.\label{eq:goal}
    \end{align}
    where we write $\oMm$ for $\Mm_{a,b,c}$ for brevity.
    
    First, let us define
	\begin{align*}
	\oMmL(u,f) \coloneqq &\ \{(F,(L,R))\in \oMm(u,f) : u\in L\},\\
	\oMmR(u,f) \coloneqq &\ \{(F,(L,R))\in \oMm(u,f) : u\in R\}.
	\end{align*}
    Clearly $(\oMmL(u,f),\oMmR(u,f))$ is a partition of $\oMm(u,f)$. Also, let
    	\begin{align*}
	\oMmL[u,f[u\mapsto \zero]] \coloneqq &\ \{(F,(L,R))\in \oMm[u,f[u\mapsto \zero]] : u\in L\},\\
	\oMmR[u,f[u\mapsto \zero]] \coloneqq &\ \{(F,(L,R))\in \oMm[u,f[u\mapsto \zero]] : u\in R\}.
	\end{align*}
    Again, $(\oMmL[u,f[u\mapsto \zero]],\oMmR[u,f[u\mapsto \zero]])$ is a partition of $\oMm[u,f[u\to \zero]]$.
    Thus,
    \begin{align}
     |\oMm(u,f)| = &\ |\oMmL(u,f)|+|\oMmR(u,f)|,\nonumber\\
     |\oMm[u,f[u\mapsto \zero]]|= &\ |\oMmL[u,f[u\mapsto \zero]]|+|\oMmR[u,f[u\mapsto \zero]]|.\label{eq:partition}
    \end{align}
    
    We now observe that by the inclusion-exclusion principle, we have
    \begin{equation}\label{eq:ie-L}
      |\oMmL(u,f)| = |\oMm[u,f[u\mapsto \twol]]| - 2|\oMm[u,f[u\mapsto \onel]]| + |\oMmL[u,f[u\mapsto \zero]]|.     
    \end{equation}
    Indeed, in partial objects $(F,(L,R))$ counted on the left hand side, both vertices $u^0$ and $u^1$ have to be incident to an edge of $F$.
    On the right hand side we count it by first only allowing both $u^0$ and $u^1$ to be incident to edges of $F$, 
    then subtracting terms corresponding to disallowing this either for $u^0$ or for $u^1$, and finally adding a correction term where both $u^0$ and
    $u^1$ are disallowed to be incident to edges of $F$. Note here that by symmetry, the two subtracted terms are equal, and equal to $|\oMm[u,f[u\mapsto \onel]]|$; hence the factor $2$.
    Analogously we argue that 
    \begin{equation}\label{eq:ie-R}
      |\oMmR(u,f)| = |\oMm[u,f[u\mapsto \twor]]| - 2|\oMm[u,f[u\mapsto \oner]]| + |\oMmR[u,f[u\mapsto \zero]]|.
    \end{equation}
    Now~\eqref{eq:goal} follows by adding equations~\eqref{eq:ie-L} and~\eqref{eq:ie-R} and using~\eqref{eq:partition}.
\end{proof}

\begin{proof}[Proof of Lemma~\ref{lem:inclusive-step}]
    We define function (here $\prod$ denotes the Cartesian product)
    $$\xi\colon \Mm[u,f]\to \prod_{v\in \child(u)} \Mm(v,f)$$ 
    as follows: for $(F,(L,R))\in \Mm[u,f]$, set
    $$\xi(F,(L,R))\coloneqq \left(\,(F\cap \sheaf[v],(L\cap \broom[v],R\cap \broom[v]))\ \colon\ v\in \child(u)\,\right).$$
    It is straightforward to verify from definitions that for each $(F,(L,R))\in \Mm[u,f]$ and $v\in \child(u)$, the pair $(F\cap \sheaf[v],(L\cap \broom[v],R\cap \broom[v]))$ belongs to $\Mm(v,f)$. Hence
    we may indeed set the co-domain of $\xi$ to be $\prod_{v\in \child(u)} \Mm(v,f)$. Also, it is clear that $\xi$ is injective.
    
    We now observe that $\xi$ is also surjective. Indeed, since
    $\{\sheaf[v]\colon v\in \child(u)\}$ is a partition of $\sheaf[u]$ by
    Observation~\ref{lem:sheaf-partition} and $(\cup_{v \in \child(u)} L_v,\cup_{v \in \child(u)} R_v)$ is compatible with $f$ since each $(L_v,R_v)$ is compatible with $f$ and $\{\subtree[v]\colon v\in \child(u)\}$ is a partition of $\subtree(u)$, 
    it is again straightforward to verify from definitions that the following assertion holds:
    If for each $v\in \child(u)$ we have some $(F_v,(L_v,R_v))\in \Mm(v,f)$, then setting
    $$F\coloneqq \bigcup_{v\in \child(u)} F_v,\qquad L\coloneqq \bigcup_{v\in \child(u)} L_v,\qquad R\coloneqq \bigcup_{v\in \child(u)} R_v$$
    yields a pair $(F,(L,R))$ that belongs to $\Mm[u,f]$ and satisfies $$\xi(F,(L,R))=(\,(F_v,(L_v,R_v))\ \colon\ v\in \child(u)\,).$$ 
    This implies that $\xi$ is a bijection between $\Mm[u,f]$ and $\prod_{v\in \child(u)} \Mm(v,f)$.
    
    Finally, we observe that by Observation~\ref{lem:sheaf-partition} and the fact that $\xi$ is a bijection, we have
    \begin{align*}
    P[u,f]& = \sum_{(F,(L,R))\in \Mm[u,f]}\, \alpha^{\omega(F)}\beta^{|F|}\gamma^{|F\cap E'_1|} \\
          & = \sum_{(F,(L,R))\in \Mm[u,f]}\ \prod_{v\in \child(u)}\, \alpha^{\omega(F\cap \sheaf[v])}\beta^{|F\cap \sheaf[v]|}\gamma^{|F\cap \sheaf[v]\cap E'_1|} \\
          & = \prod_{v\in \child(u)}\ \sum_{(F_v,(L_v,R_v))\in \Mm(v,f)}\, \alpha^{\omega(F_v)}\beta^{|F_v|}\gamma^{|F_v\cap E'_1|} \\
          & = \prod_{v\in \child(u)}\, P(v,f).
    \end{align*}
    This concludes the proof.
\end{proof}

\begin{proof}[Proof of Lemma~\ref{lem:leaf-step}]
Let $Z\coloneqq \pi(\sheaf[u])$. We claim that
\begin{equation}\label{eq:base}
P[u,f] = \prod_{xy \in Z} Q[xy,f]\cdot \prod_{x \in \tail[u]} R[x,f],
\end{equation}
where
\[
Q[xy,f] = \begin{cases}
1+ij\cdot \alpha^{\omega(xy)}\beta\gamma & \quad \textrm{if } (f(x),f(y))=(i_{\mathsf{L}},j_{\mathsf{L}}) \textrm{ or } (f(x),f(y))=(i_{\mathsf{R}},j_{\mathsf{R}}),\ i,j\in \{1,2\},\\     
1 & \quad \textrm{otherwise}, 
\end{cases}
\]
and 
\[
R[x,f] = \begin{cases}
1 + \beta & \quad\textrm{if } f(x) \in \{\twol,\twor\},\\
2 & \quad \textrm{if }f(x)=\zero,\\
1 & \quad \textrm{otherwise}. 
\end{cases}
\]
Note that this will conclude the proof, as the product~\eqref{eq:base} can be expanded into a sum of monomials in variables $\alpha,\beta,\gamma$ in polynomial time.

To argue the correctness of formula~\eqref{eq:base}, let us recall that the coefficients of $P[u,f]$ should count the pairs $(F,(L,R))$ 
that are compatible with $f$ according to Definition~\ref{def:compatible}, where we put $X=\emptyset$, $Y=\tail[u]$, and $S=\sheaf[u]$, so that
the coefficient by $\alpha^a\beta^b\gamma^c$ is the number of such pairs with $\omega(F)=a$, $|F|=b$, and $|F\cap E'_1|=c$.
We now show that each such pair $(F,(L,R))$ can be described by independent choices made for each $xy\in Z$ and each $x\in \tail[u]$, which respectively correspond to the factors in formula~\eqref{eq:base}.

For every edge $xy\in Z$, within $\pi^{-1}(xy)=\{x^0y^0,x^1y^0,x^0y^1,x^1y^1\}$ there may be at most one edge to $F$ (because $F$ needs to be simple), and this can happen only when
$(f(x),f(y))=(i_{\mathsf{L}},j_{\mathsf{L}})$ or $(f(x),f(y))=(i_{\mathsf{R}},j_{\mathsf{R}})$ for some $i,j\in \{1,2\}$. In this case, the number of possibilities for choosing the edge from $F$ is $ij$.
This explains the formula for $Q[xy,f]$: the summand $1$ corresponds to the option of not choosing any edge from $\pi^{-1}(xy)$, while the summand $ij\cdot \alpha^{\omega(xy)}\beta\gamma$ corresponds to
the option of choosing any one of the edges from $\pi^{-1}(xy)$. Note that the degrees by $\alpha,\beta,\gamma$ respectively correspond to the contribution of this edge to $\omega(F),|F|,|F\cap E'_1|$.

Next, for every vertex $x\in \tail[u]$, the edge $x^0x^1$ may be added to $F$ only when we have $f(x)\in \{\twol,\twor\}$, and in this case it contributes only to $|F|$, since $\omega(x^0x^1)=0$ and $x^0x^1\notin E'_1$.
Further, $x$ has to belong to $L$ if $f(x)\in \{\onel,\twol\}$, and $x$ has to belong to $R$ if $f(x)\in \{\oner,\twor\}$, but if $f(x)=\zero$ then we may include $x$ either in $L$ or in $R$.
This explains the formula for $R[x,f]$: the summand $\beta$ corresponds to the option of including $x^0x^1$ in $F$, while the $2$ corresponds to the two options of including $x$ either in $L$ or in $R$.

Since the choices made for different edges $xy\in Z$ and for different vertices $x\in \tail[u]$ do not restrict each other, formula~\eqref{eq:base} for $P[u,f]$ follows. 
\end{proof}

\section{Conclusion and Further Research}\label{sec:conc}

In this paper we answered the open question of Hegerfeld and Kratsch~\citet{stacs2020} by presenting
an $\Os(5^{\td})$-time and polynomial space algorithm for \textsc{Hamiltonian
Path}, \textsc{Hamiltonian Cycle}, \textsc{Longest Path}, \textsc{Longest Cycle}
\textsc{Min Cycle Cover}, where $\td$ is the depth of a provided elimination forest of the input graph.
However, there are still multiple open problems around time- and space-efficient algorithms on graphs of bounded treedepth.
We list here a selection.

\paragraph{Approximation of treedepth.} Recall that the treewidth of a graph can be approximated up to a constant factor in fixed-parameter time.
For instance, the classic algorithm of Robertson and Seymour~\cite{GM13} (see also~\cite{book1}) takes on input a graph $G$ and integer $t$, works in time $2^{\Oh(\tw)}\cdot n^{\Oh(1)}$ and in polynomial space,
and either concludes that the treewidth of $G$ is larger than $\tw$, or finds a tree decomposition of $G$ of width at most $4\tw+4$.
This means that for the purpose of designing $2^{\Oh(\tw)}\cdot n^{\Oh(1)}$-time algorithms on graphs of treewidth $t$, we may assume that a tree decomposition of approximately optimum width is given,
as it can be always computed from the input graph within the required complexity bounds.
Unfortunately, no such approximation algorithm is known for the treedepth. Namely, it is known that the treedepth can be computed exactly in time and space $2^{\Oh(\td^2)}\cdot n$~\cite{ReidlRVS14}
and approximated up to factor $\Oh(\tw\log^{3/2} \tw)$ in polynomial time~\cite{esa19}, where $\td$ and $\tw$ are the values of the treedepth and the treewidth of the input graph, respectively.
A piece of the theory that seems particularly missing is a constant-factor approximation algorithm for treedepth running in time $2^{\Oh(\td)}\cdot n^{\Oh(1)}$; polynomial space usage would be also desired.



\paragraph{Faster algorithms.}
The bases of the exponent of the running times of the algorithms given by Hegerfeld and Kratsch~\citet{stacs2020} for the treedepth parameterization match the ones obtained by
Cygan et al.~\citet{focs2011} for the treewidth parameterization.
In the case of our results, the situation is different: while {\sc{Hamiltonian Cycle}} can be solved in time $4^\tw\cdot n^{\Oh(1)}$  in graphs of treewidth $\tw$~\cite{focs2011} and
in time $(2+\sqrt{2})^p\cdot n^{\Oh(1)}$ in graphs of pathwidth $p$~\cite{cyganjacm}, we needed to increase the base of the exponent to~$5$ in order to achieve polynomial space complexity for the treedepth parameterization.
As the treedepth of a graph is never smaller than its pathwidth, it is natural to ask whether there is an $(2+\sqrt{2})^\td \cdot n^{\Oh(1)}$-time polynomial-space algorithm for {\sc{Hamiltonian Cycle}} on graphs
of treedepth~$\td$. In fact, reducing the base~$5$ to any $c<5$ would be interesting.


\paragraph{Derandomization.} 
Shortly after its introduction, the Cut\&Count technique for the treewidth parameterization has been derandomized.
Bodlaender et al.~\citet{rank-based} presented two approaches for doing~so.
The first one, called the \emph{rank-based approach}, boils down to maintaining a small set of representative partial solutions along the dynamic programming computation, and pruning irrelevant partial solutions
on the fly using Gaussian elimination. Fomin et al.~\citet{matroid-derandomization} later reinterpreted this technique in the language of matroids and extended it.
The second approach, called \emph{determinant-based}, uses the ideas behind Kirchoff's matrix-tree theorem to deliver a formula for counting suitable spanning trees of a graph, which can be efficiently
evaluated by a dynamic programming over a tree decomposition.

It seems to us that none of these approaches applies in the context of the treedepth parameterization, where we additionally require polynomial space complexity.
For the rank-based and matroid-based approaches, they are based on keeping track of a set of representative solutions, which in the worst case may have exponential size.
In the determinant-based approach, when computing the formula for the number of spanning trees over a tree decomposition,
the aggregation of dynamic programming tables is done using operations that are algebraically more involved, and which in particular are non-commutative.
See the work of W\l{}odarczyk~\cite{clifford-algebra} for a discussion. 
It is unclear whether this computation can be reorganized
so that in the aggregation we use only pointwise product --- which, in essence, is our current methodology from the algebraic perspective.


Hence, it is highly interesting whether our algorithm, or the algorithms of Hegerfeld and Kratsch~\cite{stacs2020}, can be derandomized while keeping running time $2^{\Oh(\td)}\cdot n^{\Oh(1)}$
and polynomial space usage.

\paragraph{Other graph parameters.}
Actually, Hegerfeld and Kratsch~\citet{stacs2020}~were not the first to employ Cut\&Count on structural graph parameters beyond treewidth.
Pino et al.~\citet{branch-width} used Cut\&Count and rank-based approach to get single-exponential time algorithms for 
connectivity problems parametrized by \emph{branchwidth}. Recently, Cut\&Count was also
applied in the context of \emph{cliquewidth}~\cite{clique-width}, and of \emph{$\mathbb{Q}$-rankwidth},
\emph{rankwidth}, and \emph{MIM-width}~\cite{rankwidth}.
All these algorithms have exponential space complexity, as they follow the standard dynamic programming approach.
One may expect that maybe for the depth-bounded counterparts of cliquewidth and rankwidth --- {\em{shrubdepth}}~\cite{GanianHNOM19} and {\em{rankdepth}}~\cite{devos2019branchdepth} --- 
time-efficient polynomial-space algorithms can be designed, similarly as for treedepth.

\subsection*{Acknowledgements} We would like to thank anonymous reviewers
for their suggestions and comments. The second author would like to thank Marcin Wrochna 
for some early discussions on combining treedepth and Cut\&Count.
The work leading to the results presented in this paper 
was initiated during the Parameterized
Retreat of the algorithms group of the University of Warsaw (PARUW), held in Karpacz
in February 2019. This Retreat was financed by the project CUTACOMBS, which has received
funding from the European Research Council (ERC) under the European Union’s Horizon 2020
research and innovation programme (grant agreement No. 714704).

\bibliographystyle{abbrv}
\bibliography{bib}

\begin{appendices}

\section{Faster polynomial-space algorithm for \textsc{Long Path} in $H$-minor free graphs}
\label{applications}

Let us fix a graph $H$ and consider the class of $H$-minor-free graphs, that is, graphs that exclude $H$ as a minor.
Lokshtanov et al.~\cite{k-path} gave an algorithm for the {\sc{Long Path}} problem on $H$-minor-free graphs that runs in time $2^{\Oh(\sqrt{\ell} \log^2{\ell})}\cdot n^{\Oh(1)}$ and uses polynomial space. 
We now show that this result can be slightly improved using a combination of our findings and previously known techniques, such as Turing kernelization for \textsc{Longest Path} and basic bidimensionality.

\begin{theorem}
    For every fixed graph $H$, there is a randomized algorithm for the
    \textsc{Long Path} problem on $H$-minor free graphs
    that runs in time $2^{\Oh(\sqrt{\ell} \log{\ell})}\cdot n^{\Oh(1)}$ and uses polynomial
    space. The algorithm has one-sided error: it can return false negatives with probability at most $\frac{1}{2}$.
\end{theorem}
\begin{proof}
    We first recall that Jansen et al.~\cite{turing-kernel} gave a {\em{polynomial Turing kernel}} 
    for the {\sc{Longest Path}} problem in $H$-topological-minor-free, for every fixed $H$.
    That is,
    they showed how to solve {\sc{Longest Path}} in $H$-topological-minor-free graphs by a polynomial-time algorithm that has access to
    an oracle solving the problem on $H$-topological-minor-free graphs which have $\ell^{\Oh(1)}$ vertices.
    
    Therefore, to give an algorithm for {\sc{Longest Path}} in $H$-minor-free graphs with running time $2^{\Oh(\sqrt{\ell} \log{\ell})}\cdot n^{\Oh(1)}$ and polynomial space usage,
    it suffices to implement the oracle in time $2^{\Oh(\sqrt{\ell}\log \ell)}$ and in polynomial space. That is, we need 
    to give an algorithm that achieves running time $2^{\Oh(\sqrt{\ell}\log \ell)}$ and polynomial space usage under the assumption that $n=\ell^{\Oh(1)}$.
    Note here that this algorithm can be randomized with one-sided error, as we can reduce the error probability of every oracle call to $\frac{1}{n^{\Oh(1)}}$ by repeating it $\Oh(\log n)$ times,
    so that the overall probability that any oracle call returns an incorrect answer is bounded by $\frac{1}{2}$.
    
    
    Hence, let us fix the input instance $(G,\ell)$, where we assume that $G$ is $H$-minor-free and has $n=\ell^{\Oh(1)}$ vertices.
    Similarly to Lokshtanov et al.~\cite{k-path}, we use the following observation of Demaine et al.~\cite{bidimensionality}: for every graph $H$ there
    exists a constant $h$ such that for every $p\in \N$, in every $H$-minor free graph of treewidth larger than $hp$ there is a path on at least $p^2$ vertices.
    We set $t = h\left\lceil\sqrt{\ell}\right\rceil$ and observe that this means that if the treewidth of $G$ is larger than $t$, then $G$ contains a path on $\ell$ vertices and the algorithm can return a positive answer.
      
    We can now apply  the classic approximation algorithm for treewidth of Robertson and Seymour~\cite{GM13} (see also a presentation in~\cite{book1}) to $G$ and parameter $t$.
    This algorithm runs in time $2^{\Oh(t)}\cdot n^{\Oh(1)}=2^{\Oh(\sqrt{\ell})}$ and uses polynomial space, and it either finds a tree decomposition of $G$ of width at most $4t+4$, which is bounded by $\Oh(\sqrt{\ell})$, 
    or correctly concludes that the treewidth of $G$ is larger than $t$. As we argued, in the latter case we may terminate the algorithm and return a positive answer.
    Hence, we are left with investigating the former case, when a suitable tree decomposition has been found.
    
    As shown by Bodleander et al.~\cite{bodlaender-treedepth-apx}, every tree decomposition of width $w$ can be transformed in 
    polynomial time into an elimination forest of depth $\Oh(w \log{n})$. Note that in our case $n = \ell^{\Oh(1)}$, hence the depth $d$ of this elimination forest 
    is bounded by $\Oh(\sqrt{\ell}\log{\ell})$.

    Finally, we apply the algorithm of Theorem~\ref{main-cor} for {\sc{Longest Path}} on this elimination forest.
    This algorithm runs in time $2^{\Oh(d)}\cdot n^{\Oh(1)} = 2^{\Oh(\sqrt{\ell}\log{\ell})}$ and uses polynomial space.
\end{proof}

\section{Problems Definitions}\label{sec:problems}

\defproblem{\textsc{Long Path}}
{An undirected graph $G$ and an integer $\ell$}
{Is there a simple path on $\ell$ vertices in $G$?}

\defproblem{\textsc{Long Cycle}}
{An undirected graph $G$ and an integer $\ell$}
{Is there a simple cycle of length $\ell$ in $G$?}

\defproblem{\textsc{Hamiltonian Path}}
{An undirected graph $G$}
{Is there a simple path in $G$ that visits all the vertices?}

\defproblem{\textsc{Hamiltonian Cycle}}
{An undirected graph $G$}
{Is there a simple cycle in $G$ that visits all the vertices?}

\defproblem{\textsc{Min Cycle Cover}}
{An undirected graph $G$ and an integer $k$}
{Can the vertices of $G$ be covered with at most $k$ vertex-disjoint cycles?}

\defproblem{\PCC}
{An undirected graph $G$, integers $k$ and $\ell$.}
{Is there a family of at most $k$ vertex-disjoint
cycles in $G$ that jointly visit exactly $\ell$ vertices?}

\end{appendices}

\end{document}